\documentclass[a4paper,USenglish,cleveref,numberwithinsect]{lipics-v2019}
\nolinenumbers

\usepackage[utf8]{inputenc} 
\usepackage[T1]{fontenc}    
\usepackage{url}            
\usepackage{booktabs}       
\usepackage{amsmath}        
\usepackage{amssymb}
\usepackage{amsfonts}
\usepackage{multirow}
\usepackage{multicol}
\usepackage[dvipsnames]{xcolor}         
\usepackage{nicefrac}       
\usepackage{microtype}      
\usepackage[ruled,linesnumbered,algo2e,vlined]{algorithm2e}
\usepackage{graphicx}
\usepackage{xspace}

\usepackage{hyperref}       

\RequirePackage{thmtools,thm-restate} 
\RequirePackage{refcount}             
\RequirePackage{xstring}              

\newenvironment{restated}[2]{%
  \restatable{#1}{#2}
    \label{#2}
}{%
  \endrestatable
}

\newcommand{\restate}[1]{%
  \StrBefore{\getrefnumber{#1}}{.}[\refchapter]
  \let\oldchapter\thesection
  \renewcommand{\thesection}{\refchapter}
  \csname #1\endcsname*
  \let\thesection\oldchapter
}

\newenvironment{prestated}[1]{
  \StrBefore{\getrefnumber{#1}}{.}[\refchapter]
  \let\oldchapter\thesection%
  \renewcommand{\thesection}{\refchapter}%
}{%
  \let\thesection\oldchapter%
}

\newcommand{\prestate}[1]{%
  \csname #1\endcsname
}

\makeatletter
\newcommand\appendix@section[1]{%
\refstepcounter{section}%
\orig@section*{Appendix \@Alph\c@section: #1}%
\addcontentsline{toc}{section}{Appendix \@Alph\c@section: #1}%
}
\let\orig@section\section
\g@addto@macro\appendix{\let\section\appendix@section}
\makeatother

\newcommand{\current}{\textnormal{\texttt{current}}\xspace}
\newcommand{\position}{\textnormal{\texttt{position}}\xspace}
\newcommand{\permit}{\textnormal{\texttt{permit}}\xspace}
\newcommand{\permits}{\textnormal{\texttt{permits}}\xspace}
\newcommand{\Proof}{\textnormal{\texttt{proof}}\xspace}
\newcommand{\block}{\textnormal{\texttt{block}}\xspace}
\newcommand{\proposal}{\textnormal{\texttt{proposal}}\xspace}
\newcommand{\timeout}{\textnormal{\texttt{timeout}}\xspace}

\newcommand{\prot}{PermitBFT\xspace}

\title{\prot: Exploring the Byzantine Fast-Path}

\author{Roland Schmid}{ETH Zürich, Switzerland}{roschmi@ethz.ch}{}{}
\author{Roger Wattenhofer}{ETH Zürich, Switzerland}{wattenhofer@ethz.ch}{}{}

\authorrunning{R.\ Schmid and R.\ Wattenhofer}



\hideLIPIcs  

\EventEditors{John Q. Open and Joan R. Access}
\EventNoEds{2}
\EventLongTitle{42nd Conference on Very Important Topics (CVIT 2016)}
\EventShortTitle{CVIT 2016}
\EventAcronym{CVIT}
\EventYear{2016}
\EventDate{December 24--27, 2016}
\EventLocation{Little Whinging, United Kingdom}
\EventLogo{}
\SeriesVolume{42}
\ArticleNo{23}

\begin{document}

\maketitle

\begin{abstract}
    \prot establishes a permissioned byzantine ledger in the partially synchronous networking model. For $n$ replicas, \prot tolerates up to $f < \frac{n}{3}$ byzantine replicas. It is the first BFT protocol to achieve a latency of just 2 message delays despite tolerating byzantine replicas throughout the ``fast track'', as long as they are not the leader.
    The design of \prot relies on two fundamental concepts. First, in \prot the participating nodes do not wait for a distinguished leader to act and subsequently confirm its actions, but send permits to the next leader proactively. Second, \prot achieves a separation of the decision powers that are usually concentrated on a single leader node. A leader in \prot controls \textit{which} transactions to include in a new block, but not \textit{where} to append the block in the block graph.
\end{abstract}

\section{Introduction}
\label{sec:introduction}
A distributed system consists of $n$ nodes. The system is byzantine fault tolerant (BFT) if it can tolerate at most $f < \frac{n}{3}$ arbitrarily malicious (\textit{byzantine}) nodes. BFT protocols have been studied in great detail since many decades, both in theory and practice. Nowadays, 
BFT protocols are the key to building ``permissioned blockchains'', an area traditionally known as  
``state machine replication''~\cite{lamport1984using,schneider1990implementing}.

\paragraph*{Permissioned Byzantine Ledger Without Consensus}{
    In practice, BFT protocols have many applications ranging from online shopping to credit card transactions, cryptocurrencies and stock market trades; whenever a set of clients makes concurrent requests for (or with) limited resources, the service providers have an interest to both prevent fraudulent and tolerate faulty behaviour in the system. While BFT systems have traditionally implemented consensus algorithms, recent works have shown that this is not needed for many applications such as cryptocurrencies~\cite{guerraoui2019consensusnumber}.
    While some of these applications may not require to establish a total ordering at all~\cite{baudet2020fastpay,sliwinski2019abc}, other applications may need to settle on a total ordering of committed requests eventually, but still provide constraints on the incoming requests.

    Originally, the interest in BFT systems has first been reignited by Castro and Liskov, when they presented their ``Practical'' BFT (PBFT) system \cite{castro1999practical}. After PBFT, a large number of other BFT systems emerged. These systems try to minimize the delay until transactions are committed. They do not optimize the worst case, since the worst case happens rarely in practice. Rather, they minimize the delay with a varying degree of optimistic assumptions (e.g., no message timeouts, leader not byzantine, no node byzantine, transactions already pre-ordered).
    Therefore, we set out to design a BFT protocol that is more resilient to byzantine faults and yet matches the commit latency of the so-called ``fast path'' operation modes of other BFT protocols.
}

\paragraph*{Byzantine-Resilient Latency}{
    To provide a practical (yet purely theoretical) evaluation of different BFT protocols with regard to their respective latencies without focusing entirely on worst-case scenarios, we consider the byzantine-resilient latency of a BFT protocol:
    \begin{definition}
        \label{def:optimistic-latency}
        We start the clock when the leader receives a transaction, and we stop the clock when this transaction is for sure included in the ledger. Local computations are free. The leader is timely and not byzantine. At most $f$ backup nodes are byzantine. Nodes may see messages in a different order, we only assume that messages take at most 1 time unit.
    \end{definition}
    \begin{table}[t]
        \centering
        \begin{tabular}{lcccc}
            \toprule
            \multirow{2}{*}{Protocol} & Best-Case & Byzantine-resilient & \multicolumn{2}{c}{Message complexity} \\
             & latency & latency & \textit{Normal operation} & \textit{(Leader) failure} \\
            \midrule
            PBFT \cite{castro1999practical} & $3$ & $3$ & $\mathcal{O}(n^2)$ & $\mathcal{O}(n^3)$ \\
            Zyzzyva \cite{kotla2007zyzzyva} & $2$ & $4$\makebox[0pt][l]{${}^{{}^*}$} & $\mathcal{O}(n)$ & $\mathcal{O}(n^2)$ \\
            Tendermint \cite{buchman2018latest} & $3$ & $3$ & $\mathcal{O}(n^2)$ & $\mathcal{O}(n^2)$ \\
            SBFT \cite{gueta2019sbft} & $2$ & $4$ & $\mathcal{O}(n)$ & $\mathcal{O}(n^2)$ \\
            HotStuff \cite{yin2019hotstuff} & $6$ & $6$\makebox[0pt][l]{${}^{{}^{**}}$} & $\mathcal{O}(n)$ & $\mathcal{O}(n)$ \\
            PaLa \cite{chan2018pala} & $4$ & $4$ & $\mathcal{O}(n)$ & $\mathcal{O}(n^2)$\makebox[0pt][l]{${}^{{}^{***}}$} \\
            \textbf{\prot} & $\boldsymbol{2}$ & $\boldsymbol{2}$ & $\boldsymbol{\mathcal{O}(n)}$ & $\boldsymbol{\mathcal{O}(n^2)}$\makebox[0pt][l]{${}^{{}^{***}}$} \\
            \toprule
        \end{tabular}
        \begin{scriptsize}
            ${}^{{}^*}$Zyzzyva
            's speculative latency of 2 rounds may not be achieved with a single crash failure.\\
            ${}^{{}^{**}}$HotStuff employs a pipelining scheme optimizing for throughput rather than latency.\\
            ${}^{{}^{***}}$In practice, the message complexity may be improved to $\mathcal{O}(n)$ with the ideas outlined in \Cref{sec:message-complexity}.
        \end{scriptsize}
        \medskip
        \caption{Evaluation of the \textit{byzantine-resilient latency} of selected practical BFT protocols, plus their message complexity.}
        \label{tab:optimistic-latency}
    \end{table}
    An overview of achieved byzantine-resilient latencies can be found in \Cref{tab:optimistic-latency}.
    In this paper we present \prot, a new BFT protocol to build a permissioned byzantine ledger. The \prot protocol inverts the process of previous BFT protocols, achieving a better bound for transaction commit delay in the presence of byzantine backups. In a nutshell, \prot works as follows.
}

\paragraph*{Protocol Description}{
    We consider a system of $n$ authenticated nodes, i.e.\ all messages are signed.
    Initially, every node stores the genesis block that contains an ordering of the participating nodes.
    
    All nodes permit a \textit{leader} (each node is selected round-robin) to append a new block at a specific position in the block graph. While the position of the new block is given by the permits, a leader may freely choose to include transactions in a new block. If the leader creates a block in a timely manner, the nodes will allow the next leader to append to that block. This can be seen as the nodes voting on the position \textit{where} to append the next block, thus ensuring that the block graph maintains its desired structural properties (details in \cref{sec:analysis}).
    
    In reality, however, there might be disagreement among the nodes on where to append a new block. There are essentially two reasons for this: either a leader did not produce the next block in a timely manner (from the perspective of sufficiently many nodes), or a byzantine leader created multiple blocks at the same position, a so-called \textit{fork}.
    
    If this happens, the \prot algorithm will unify the system by allowing blocks to be created with multiple parent blocks utilizing a 2-step process. Strictly speaking, \prot does not establish a chain of blocks, but rather a directed acyclic graph of blocks with a total ordering.
}



\paragraph*{Contribution}{
\prot establishes a permissioned byzantine ledger that provides a total ordering of committed requests and incurs a byzantine-resilient latency of only 2 message delays (1 RTT). To the best of our knowledge, it is the first BFT protocol to achieve this combination of properties.
\prot does, however, not solve byzantine consensus as not all incoming requests may be ordered (and also not committed). Hence, it depends on the application context whether this shortcoming is relevant or not. For example, when building a digital currency, this would only impact malicious actors and thus was not a concern in the protocol design.
}

\section{Related Work}
\label{sec:related-work}

Byzantine fault-tolerant (BFT) systems have been studied for several decades under the names of byzantine consensus and state machine replication~\cite{lamport1982byzantine,fischer1985impossibility,dwork1988consensus,castro1999practical}. Recently, the area regained traction motivated by the needs of permissioned blockchain~\cite{gueta2019sbft,yin2019hotstuff}. Unlike ``proof of work''-style blockchains such as Bitcoin~\cite{nakamoto2008bitcoin} causing high energy consumption, permissioned blockchains aim to leverage the known identities of the participants to achieve agreement based on authentication.

At the same time, permissioned blockchains can provide instant finality of a reached decision, providing indisputable commit certificates for each included transaction. Other than eventually-consistent blockchains, there simply cannot be another (conflicting) quorum of signatures reached as long as the fault-tolerance and cryptographic assumptions uphold.
For \textit{practical} BFT systems~\cite{castro1999practical}, the achievable commit latency for client requests is as low as one round-trip time (1 RTT), in an optimistic setting when neither network nor node failures occur~\cite{kotla2007zyzzyva,gueta2019sbft,song2008bosco}. More specifically, Zyzzyva~\cite{kotla2007zyzzyva} focuses on the idea of a speculative fast track. This can be motivated by the observation that, in practical systems, reliable systems will only rarely experience faulty or malicious behaviour of the participating nodes. However, even a single crash failure voids this advantage immediately. While this was addressed in SBFT~\cite{gueta2019sbft}, they similarly only tolerate a constant number of crash failures in the fast path. In contrast, with \prot, we propose a protocol which achieves a byzantine-resilient latency of 1~RTT when only the current leader is assumed to be correct and the network is stable. In other words, up to $f < \frac{n}{3}$ backup replicas may be faulty and exhibit even byzantine behaviour.

In all fairness, there are several recent proposals achieving a similar byzantine-resilient latency. For example, both the proof-of-stake protocol ABC~\cite{sliwinski2019abc} or the FastPay protocol~\cite{baudet2020fastpay}. Similar to our own proposal, these systems do not solve consensus. Even more so, these systems do not even establish a total ordering of transactions at all. Nonetheless, the proposed setting exhibits great relevance in matching the circumstances found in transaction-based payment systems perfectly. More precisely, if a client knows all of its previously signed transactions and can thus easily provide the most recent transaction to the validator set, these systems only require the validators to verify a partial ordering of all their signed transactions. 
While \prot may also fail to order (and thus commit) conflicting client requests, it provides a total order of committed requests and is thus a suitable alternative for order-critic applications such as limited-resource allocation (e.g.\ whoever pays first wins).

Other attempts to increase the byzantine fault-tolerance in the ``normal operation'' mode of BFT protocols include Prime~\cite{amir2010prime} and FnF-BFT~\cite{avarikioti2020fnf-bft}; however, neither of the proposals evaluate the theoretical commit latency of a request and are thus orthogonal to \prot.



\section{Model}
\label{sec:model}

    We consider a distributed system consisting of $n$ nodes and a number of external clients issuing requests to the nodes.
    The \prot algorithm tolerates up to $f < \frac{n}{3}$ \textit{byzantine failures}, that is, malicious adversarial nodes; all other nodes are assumed to be \textit{correct} nodes following the protocol.
    Nodes must be authenticated to participate in the protocol, meaning that all nodes possess a cryptographic public-private key pair suitable for signing messages where the public keys are known to all nodes.
    We assume a Bitcoin-style UTXO transaction system \cite{delgado2018utxo}, in which a client's signature is sufficient to legitimate executing a client request.
    
    As an adversarial model, we assume that byzantine nodes may communicate arbitrarily among each other, thus acting as if they were controlled by a single global adversary with unified information. Byzantine nodes' computational power, however, is assumed to be polynomially bounded. In particular, we assume that byzantine nodes cannot forge cryptographic signatures of correct nodes.
    
    We demonstrate the safety of the \prot algorithm in the \textit{asynchronous model}, that is, assuming that messages will be delivered in a finite but unbounded time in arbitrary order. The liveness of the system is demonstrated in the \textit{partially synchronous model} \cite{dwork1988consensus} where synchronous (i.e.\ bounded-delay) periods may be assumed to occur eventually.

\paragraph*{Protocol Goals}{
    The \prot algorithm allows $n$ nodes to agree on an ordered log of transactions.
    Every node outputs a sequence of committed transactions $(pos,\ tx)$ with $pos \in \mathbb{N}$ in strictly incremental order, thereby establishing a permissioned ledger. On that ledger, we record signed transactions that are supplied by external clients.
    Two transactions can be either \textit{unrelated}, \textit{dependent}, or \textit{conflicting}. These predicates should be defined in consideration of the application context and will therefore be used as wild-cards in their intuitive sense.
    
    More formally, the \prot algorithm establishes a \textit{permissioned byzantine ledger}. A permissioned byzantine ledger satisfies the following properties:
    \begin{itemize}
        \item \textbf{Total Order:} If any correct node outputs transaction $tx_1$ before $tx_2$, then no correct node outputs $tx_2$ before $tx_1$.
        \item \textbf{Agreement:} All correct nodes commit the same transaction $tx$ at position $pos \in \mathbb{N}$.
        \item \textbf{Finality:} For every transaction $tx$ that is committed by at least one correct node, there is a certificate that proves the transaction will remain incontestably committed.
        \item \textbf{Inviolability:} No node can obtain a certificate to prove the commit of a transaction conflicting with any other committed transaction.
        \item \textbf{Liveness:} A transaction issued by a correct client that is received by at least one correct node will eventually be committed on the permissioned byzantine ledger.
    \end{itemize}
}

\section{\prot}

In this section, we first give a general overview of the \prot protocol and subsequently discuss the individual steps of the protocol in more detail. In addition, the \prot protocol is given in pseudo-code in \cref{alg:permit-bft}. Note that some implementation details were omitted for simplicity; for instance, nodes need to collect transactions and fetch unknown blocks (that are referenced by a \permit, \block or \proposal) asynchronously in the background, and \timeout messages will be answered if a node has already progressed to a newer round (see \Cref{sec:protocol:timeouts-failures}).

\SetKwInput{Input}{Input }
\SetKwInput{Output}{Output }
\SetNlSty{texttt}{}{:}
\SetAlFnt{\tt}
\newcommand{\textbftt}[1]{\textbf{\texttt{#1}}}
\SetKwSty{textbf}
\SetArgSty{texttt}
\SetFuncArgSty{texttt}
\SetProgSty{texttt}
\SetCommentSty{textnormal}
\DontPrintSemicolon
\SetKwIF{If}{ElseIf}{Else}{if}{$\!\!\!\!\!\;$\textnormal{\texttt{:}}}{else if}{else}{end if}
\SetKwSwitch{Switch}{Case}{Other}{switch}{$\!\!\!\!\!\;$\textnormal{\texttt{:}}}{case}{otherwise}{end case}{end switch}
\SetKwFor{ForAll}{forall}{$\!\!\!\!\!\;$\textnormal{\texttt{:}}}{end forall}
\SetKwFor{For}{for}{$\!\!\!\!\!\;$\textnormal{\texttt{:}}}{end for}
\SetKwFor{While}{while}{$\!\!\!\!\!\;$\textnormal{\texttt{:}}}{end while}
\SetKwIF{Wait}{WaitElseIf}{WaitElse}{wait}{$\!\!\!\!\!\;$\textnormal{\texttt{:}}}{wait}{wait}{end wait}
\SetKwIF{Unless}{UnlessElseIf}{UnlessCondition}{unless}{$\!\!\!\!\!\;$\textnormal{\texttt{:}}}{unless}{unless}{end unless}
\SetKwIF{RepeatFor}{RepeatElseIf}{Repeat}{repeat}{$\!\!\!\!\!\;$\textnormal{\texttt{:}}}{repeat}{repeat}{end repeat}
\SetKwIF{UntilIf}{Until}{UntilCondition}{until}{$\!\!\!\!\!\;$\textnormal{\texttt{:}}}{until}{until}{end repeat}
\SetKwIF{BreakIf}{BreakElseIf}{BreakCondition}{break if}{$\!\!\!\!\!\;$\textnormal{\texttt{:}}}{else break if}{break if}{end break if}
\SetKw{Break}{break}
\LinesNotNumbered

\begin{algorithm2e}
  \thispagestyle{empty}
  \Input{\begin{tabular}[t]{l@{ : }l}
    genesis & \textnormal{initial block of the block graph} \\
    nodes         & \textnormal{ordered set of \texttt{n} nodes}
  \end{tabular}}
  \Output{\begin{tabular}[t]{l}
    \textnormal{$\!\!\!\!$block graph with a unique set of totally-ordered committed blocks} 
  \end{tabular}}

  \medskip
  
  round = 0


  current = genesis \tcp*{position in the block graph for which to issue a \texttt{permit}}

    
  timeouts = \{$\,$\} \tcp*{dictionary mapping a \texttt{round $\mapsto$ set of timeout messages}}
    

  \Repeat{
  
    
    leader = round $\mathsf{mod}$ n \tcp*{select the round-robin leader for the round}
    
    
    Send permit(round, current) to the leader \label{alg:send-permit}\;
    
    
    \If{you are the leader}{
      
        
        permits = \{$\,$\} \tcp*{dictionary mapping a \texttt{position $\mapsto$ set of permits}} 
        
      
        \Repeat{
        
            
            Store newly received permits in the permits dictionary\;
            
            
            \If(\label{alg:enough-permits-block}){there is a position with at least 2f+1 permits}{
            
                proof = 2f+1 permits for position\;
                Send block(proof, position, transactions) to all nodes \label{alg:create-block}\;
                \Break
                
            }
            
            
            \If(\label{alg:creator-timeout}){a leader timeout occurs}{
                
                \If(\label{alg:enough-permits-proposal}){permits contains at least 2f+1 permits}{
                
                    position = minimal position that respects all permits\;
                    Send proposal(position, permits) to all nodes \label{alg:issue-proposal}\;
                }
                
                \Break\;
            }
            
            
            \If{received a block, proposal or 2f+1 timeout messages for any following round}{
                \Break \tcp*{must be behind, abort the leader phase}
            }
        }
    }


    \Repeat(\label{alg:repeat-collect-proposal-block}){
    
        
        \If(\label{alg:receive-result}\label{alg:receive-block}\label{alg:receive-proposal}){received a result (block or proposal) with result.round $\geq$ round}{
            
            round = result.round \label{alg:ffw-round-result} \tcp*{\makebox[0pt][l]{1. fast-forward to the latest round}\hphantom{3. stop accepting results for round}} 
            
            current = result.position \label{alg:update-current} \tcp*{\makebox[0pt][l]{2. update current position}\hphantom{3. stop accepting results for round}}
            
            \Break \tcp*{3. stop accepting results for round}
        }
        
        
        \If(\tcp*[f]{ensure round synchronization}){a round timeout occurs}{
        
            Send timeout(round) to all nodes \label{alg:send-timeout}\;
            
        }
        
        
        Store newly received timeout messages in the timeouts dictionary\;
        
        
        \If(\label{alg:enough-timeouts}){there is any following round with at least 2f+1 timeout messages}{
        
            round = maximum of all such rounds \label{alg:ffw-round-timeouts}\;
        
            Send 2f+1 timeout messages for round to all other nodes \label{alg:send-timeouts}\;
            
            \Break\;
            
        }
    }
    
    round = round + 1 \tcp*{end round}
    
  }
  
  
  \caption{\prot Algorithm}
  \label{alg:permit-bft}
\end{algorithm2e}

\paragraph*{Protocol Overview}{
    In \prot, the nodes agree on a common directed, acyclic \textit{block graph}, that is, they jointly create blocks with references to their parent blocks determining the graph topology. Initially, the system starts with a \textit{genesis block} that is known to all nodes and stores an ordering of the participating nodes, as well as all their public keys.
    
    In order to create new blocks, all nodes send \textit{permits} to a \textit{leader}, thereby permitting it to append a new block at a specific \textit{position} (i.e.\ its parent blocks) in the block graph. While the position of the new block is predetermined by the received permits, a leader may freely choose to include any non-conflicting transactions in the new block.
    
    If the leader creates a block in a timely manner, the nodes will issue new permits for the leader to append another block, thereby forcing it to append only to that newly created block. If the leader fails or progresses too slowly, the nodes issue \textit{timeout} messages and eventually move to the next leader. If there is disagreement about the position where the leader should append a new block, \prot allows a leader to \textit{propose} the creation of a block with multiple parent blocks, which allows the subsequent leader to unify the system.
    %
}

\paragraph*{Block Creation}{
    The \prot protocol is initialized with a genesis block that contains a fixed ordering of the participating nodes. All other blocks must be created by combining $(n-f)$ so-called \textit{permits} issued by the participating nodes.
    \begin{definition}
        \label{def:permit}
        A \permit is a tuple $\mathtt{(round,\ position)}_{\sigma_{issuer}}$ signed by its issuing node.
    \end{definition}
    Intuitively, a permit certifies that a node has endorsed the leader of the round to create a new block at the specified position.
    Each permit is bound to a specific position and block creator and permits can only be combined if they are issued for the same position and block creator. 
    \begin{definition}
        \label{def:proof}
         A \Proof is a set of $2f + 1$ permits (from distinct nodes) for the same position and block creator.
    \end{definition}
    In other words, a \Proof guarantees that the creator has been endorsed by a quorum of $2f + 1$ nodes for the specified position. As there are only $f$ byzantine nodes, observe that this must include at least $f + 1$ correct nodes, that is, a correct majority. Hence, a creator that receives $2f + 1$ permits for the same position is allowed produce a new block.
    \begin{definition}
        \label{def:block}
        A \block is a tuple $\mathtt{(proof,\ transactions)}_{\sigma_{creator}}$ signed by the creator that is determined by the proof.
    \end{definition}
    %
    Blocks contain an ordered set of non-conflicting transactions that each block creator collects at its own discretion. Correct nodes will simply add all received transactions in FIFO order to accomplish fairness.
    In this way, \prot establishes a directed acyclic graph of blocks rooted in the genesis block. However, as we want to obtain a chain of committed blocks -- a \textit{blockchain} -- blocks may not be created at arbitrary positions of the graph.
}

\paragraph*{Commit Rule}{
    Intuitively, the \prot algorithm needs to ensure that the graph of blocks allows to deduce a unique, totally-ordered path of committed blocks.
    \begin{definition}
        A block is said to be \textbf{committed} if there exists at least one child block. 
    \end{definition}
    We will show in \Cref{sec:analysis} that the structure of the block graph maintained by the \prot protocol indeed allows to deduce a unique, totally-ordered set of committed blocks, thus establishing a blockchain.
    Optimistically, a single round of the protocol proceeds as follows.
    \begin{enumerate}
        \item Whenever a new block is received, nodes issue permits to the next block creator to append from the new block.
        \item When the next creator collects a quorum of $2f + 1$ \permits (i.e.\ obtains a proof), it creates a new block including all non-conflicting transactions that it received so far.
    \end{enumerate}
    \begin{figure}[b]
        \centering
        \includegraphics{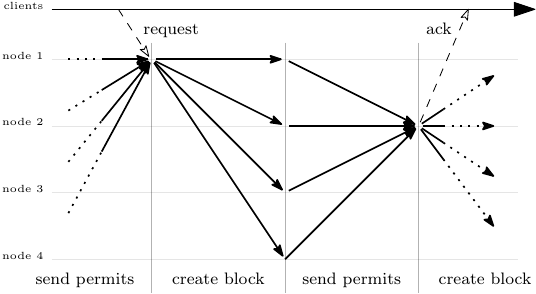}
        \caption{Schematic optimistic execution of the \prot algorithm, processing a new request.}
        \label{fig:message-schematic}
    \end{figure}
    In the best case, this procedure establishes a chain of blocks that may be interpreted as a linear log of transactions.
}

\paragraph*{Accepting New Blocks}{
    A correct node decides whether to accept a received block based on its local \current head of the block graph, that is, based on the most recent accepted position the node has stored locally. A new block can be accepted if it respects the \current position entirely:
    \begin{definition}
        \label{def:respects}
        We say that a position \textbf{respects} all blocks that lie on any path from the genesis block to the position. Furthermore, a position \textbf{respects} all 
        blocks that are in conflict with a committed block that is respected by the position (e.g.\ all uncommitted blocks that have depth at most position.depth - 2, counting from the genesis block).
        A block/permit/proposal \textbf{respects} all blocks that are respected by its associated position.
    \end{definition}
    In other words, a correct node will ensure that either its vote (in the form of a permit) is respected, or provably uncommittable.
}

\paragraph*{Block Creators}{
    In \prot, the nodes take turns creating new blocks by selecting a \textit{leader} as the current \textit{block creator}. The order in which they select the next leader is fixed in the genesis block, which is known to all participants from the start of the protocol; namely
    $$\mathtt{leader\ =\ round}\ \mathsf{mod}\ \mathtt{n}.$$
%
    In contrast to other BFT protocols, a leader in \prot does not propose the contents of a new block and subsequently collect signatures for its proposal. Instead, the backup nodes proactively issue permits for the current leader, allowing it to append any new block at the specified position in the block graph. In this way, the quorum of $(n - f)$ permits required to create a new block can be completed while simultaneously receiving new transactions, that can immediately be included into the block.
}

\paragraph*{Separation of Powers}{
    Put differently, a leader can freely decide \textit{which} transactions to include into a block it creates, and potentially even create multiple blocks -- a ``fork'' -- at the very same position in the block graph. However, each new block is bound to extend from the position specified by the set of collected permits. Hence, the leader does not have control \textit{where} to append a new block in the block graph, i.e.\ which previous blocks to commit, as the community of correct nodes remains in charge of the structure of the resulting block graph.
}

\paragraph*{Optimistic Operation}{
    Note that, optimistically, a non-conflicting transaction can be incorporated in a committed block within 2 rounds of communication (cf.\ \Cref{fig:message-schematic}). In the first round, it may already be included in a new block that is distributed to all nodes by the block creator. In the second round, the backup nodes issue permits that allow the subsequent block creator to create a child block, which may serve as an $ack$ (acknowledgment) message to the client. Hence,
    \begin{corollary}
        \prot attains a byzantine-resilient latency of 2 communication rounds.
    \end{corollary}
    Note that we assume transactions to be non-conflicting, which we ensure by employing a UTXO system.
}

\paragraph*{Resolving Disagreement}{
    \label{sec:protocol:disagreement}
    In reality, however, due to the participation of byzantine nodes and possibly unbounded network delays experienced in an asynchronous network, a few more scenarios have to be considered. For example, a byzantine creator may not create a \block at all, not distribute the created \block to all correct nodes, or even create several blocks at the specified \position in the block graph using the same \Proof.
    
    In any case, given the commit rule above that solely relies on the existence of a child \block, correct nodes cannot carelessly issue permits for different positions in each round. Correct nodes do, however, maintain a preferred position, their respective \current position, which will be updated in case of disagreement. The solution presented by \prot is to allow blocks to be created with multiple parent blocks in such a case.
    \begin{definition}
        \label{def:position}
        A \position is a set of blocks in the block graph.
    \end{definition}
    Note that this definition is required to allow blocks with multiple parent blocks, while remaining coherent with the use of the word position in the previous \Cref{def:permit,def:proof,def:respects}.
    
    Intuitively, in line with the commit rule, a \block respects all its ancestor blocks and blocks that have no chance of ever being committed. Now, for a block creator to demonstrate that there is disagreement among the nodes, we introduce so-called proposals:
    \begin{definition}
        \label{def:proposal}
        A \proposal is a tuple $\mathtt{(position,\ permits)}$ where \position is the smallest \position that respects all the positions with an issued \permit and \permits is a set of at least $2f + 1$ permits from distinct nodes.
    \end{definition}
    As soon as a node receives a valid \proposal, it updates its \current position by adopting the associated \position of the \proposal. In other words, \prot unifies the nodes' \current positions by confirming all positions that may possibly ever become committed by a child \block. We show in \Cref{sec:safety} that correct nodes can do so safely.
}

\paragraph*{Timeouts \& Failures}{
    \label{sec:protocol:timeouts-failures}
    If possible, a correct leader should always prefer to create a \block over a \proposal, as new blocks will ultimately drive progress. However, we can neither guarantee that a leader will receive the correct nodes' \permits first, nor that the leader will receive more than $n - f = 2f + 1$ permits in any given round due to $f$ possible byzantine failures. Therefore, \prot employs a \texttt{leader timeout} up to which a leader will wait before potentially issuing and broadcasting a \proposal.
    
    On the other hand, from the perspective of nodes that are not the designated leader of the round, it may be difficult to determine the end of a round. To that end, a (longer) \texttt{round timeout} is used. In order to guarantee sufficient synchronization among the nodes, each node broadcasts a \timeout message upon expiry of their local \texttt{round timeout}. A collection of $2f + 1$ \timeout messages is required to proceed to the next round without observing either a \block or a \proposal from the designated leader of this (or any following) round.
    
    Note that a byzantine leader or bad networking conditions may cause some correct node to incur a local \texttt{round timeout}, while other correct nodes did receive a \block or \proposal for the round. Thus, if a correct node receives a \timeout message for an ``old round'' (any round number smaller than its local \texttt{round} variable), the correct node answers by relaying the latest \block, \proposal, or $2f + 1$ \timeout messages that it has seen locally. Note that this has been omitted in \Cref{alg:permit-bft} for simplicity.
}

\paragraph*{Inferring the Transaction Log}{
    \label{sec:protocol:transactions}
    Similar to the Bitcoin protocol~\cite{nakamoto2008bitcoin}, we assume an ``Unspent Transaction Output'' (UTXO) \cite{delgado2018utxo} scheme which allows us to assume that conflicting transactions may only occur in case of malicious behavior. Note that two conflicting transactions may be included in different blocks that could become committed simultaneously (but not independently, see \Cref{sec:safety:transactions-with-conflicts}). To that end, let us define when a transaction itself is considered to be committed:
    \begin{definition}
        \label{def:commit-rule:transactions}
        A transaction is said to be \textit{committed} if it is contained in a committed \block~$b$ and there exists a child \block~$b'$ (a witness of the commitment of $b$) such that there are no conflicting transactions within the set of committed blocks that are respected by the child \block~$b'$.
    \end{definition}
    If any two transactions are conflicting, we cannot guarantee that none of these transactions have been committed by a (possibly hidden) \block. Hence, we can neither choose a single one of the transactions to be executed, nor can we drop the conflicting transactions entirely. Consequently, using the fact that conflicting transactions may only occur in case of malicious behavior, we \textit{freeze} the transactions. In the context of a cryptocurrency, for instance, this would imply locking the maximum amount spent across the set of pairwise conflicting transactions. This is necessary in order to guarantee that a committed transaction may be executed once its proof of commitment, a child \block only respecting a single one of the conflicting transactions, is revealed. In the worst case, this may lead to deadlocked funds; hence, immediately penalizing the transaction owner who acted maliciously.
    
    Note that the set of blocks that are respected by a \block at depth~$d$ is immutable. To see this, observe that a majority of correct nodes must have issued \permits for a position at depth~$d-1$ in order to create the \block. Hence, no new blocks can be created at depths smaller than $d-1$.
    \begin{corollary}
        Once a transaction is committed, it is irrevokably committed.
    \end{corollary}
    Hence, as soon as a transaction is committed, it may be safely executed (see \Cref{sec:safety:transactions-with-conflicts} for details). So, how do we infer a linear, totally-ordered log of transactions from the resulting block graph? \prot includes transactions into the ledger in two steps. At first, it is ensured that no two conflicting transactions become committed simultaneously which makes them safe for execution. In a second step, the definitive ordering in the ledger is finalized.
    
    In other words, \prot creates a totally-ordered ledger with a set of unordered, non-conflicting committed transactions at its head. These transactions are yet to be put into their definitive order but can already be executed.
}

\section{Security Analysis}
\label{sec:analysis}


\subsection{Safety}
\label{sec:safety}
As for every distributed system, it is crucial to guarantee the safety of the \prot algorithm independent of both the state of the system and any networking assumptions (i.e.\ in the asynchronous model). By \textit{safety} of the \prot algorithm, we denote that two conflicting transactions may not become committed simultaneously. To that end, we first analyze the structural properties of the block graph and show that there cannot be any two conflicting but independently committed blocks. Subsequently, in \Cref{sec:safety:transactions-with-conflicts}, we demonstrate that this property ensures the transaction safety of the \prot algorithm.

Formally, to accommodate for arbitrary networking conditions, we base our analysis on the \textit{asynchronous model} throughout this section. That is, we assume that messages will be delivered within a finite but unbounded time, without message loss. In particular, the order in which messages are received may be arbitrary (even from the same sender).

Intuitively, the safety of the \prot algorithm is maintained by the correct nodes who do not switch their respective \current position carelessly. Each node, individually, will decide on one of four possible results for each round before progressing to any following round: either accepting a \block or a \proposal
, skipping to the next round with a \timeout
, or by fast-forwarding to a following round
.

We thus show that a node may safely accept a \block, a \proposal or transition to the next round without updating its \current position. The basic idea of these proofs will be that either a newly accepted \block (\proposal) was supported by at least $f + 1$ correct nodes, which makes it safe, or it was supported by less than $f + 1$ correct nodes, which makes it harmless to the safety of the system as no quorum of $2f + 1$ \permits could then be achieved.

\subsubsection{Safety of Blocks}
To begin with, we consider the main backbone of the \prot algorithm: the blocks. 
\begin{definition}
    A \block is said to be \textit{promised}, if in any round up to (and including) the latest round, at least $f + 1$ correct nodes issued a permit for a common position containing the \block.
\end{definition}
Note that this is a required criterion for a block to become committed. Hence,
\begin{corollary}
    \label{lem:protocol:committed-only-if-promised}
    If a \block is committed, it is also promised.
\end{corollary}
\begin{proof}
    To commit a \block, a child \block must be created. Let round~$i$ be the first round when such a child \block was created. Thus, there must have been $2f + 1$ \permits for a particular position containing the \block in round~$i$. Given that at most $f$ of these \permits may be issued by byzantine nodes, at least $f + 1$ correct nodes must have issued a \permit for the same particular position in round~$i$. Hence, the \block is promised.
\end{proof}
With this, we can now specify when a round is safe.
\begin{definition}
    \label{def:safe-permit}
    A \textit{safe} \permit is a \permit for a position respecting all promised blocks. 
    A round is called \textit{safe} if at least $f + 1$ correct nodes would only cast a safe \permit.
\end{definition}
Note that a safe round is defined based on the potentially issued \permits of correct nodes. This way, we can reason about rounds in which the network synchronization does not allow us to assume that a correct node will actually issue a \permit for some round~$i$. However, to produce a valid \block or \proposal, any creator node must have received sufficiently many, that is, at least $2f + 1$ \permits (whereof at least $f + 1$ must be correct \permits).

\begin{lemma}
    \label{lem:safety:blocks-respect-promised}
    If round~$i$ is safe, a \block created in round~$i$ respects all promised blocks.
\end{lemma}
\begin{proof}
    A block must be created with $2f + 1$ \permits from the same round~$i$ for the exact same position~$p$
    . Hence, there must have been at least $f + 1$ correct \permits for position~$p$ in round~$i$. Given that at least $f + 1$ correct nodes would have only issued a safe \permit in round~$i$, we conclude that position~$p$ respects all promised blocks.
\end{proof}
Furthermore, we show that the safety invariant is maintained if at least $f + 1$ correct nodes accept the same \block.
\begin{lemma}
    \label{lem:safety:safely-accept-block}
    If round~$i$ is safe and a \block is accepted by at least $f + 1$ correct nodes, then round~$i + 1$ is safe.
\end{lemma}
\begin{proof}
    If there is a \block~$b$ created in round~$i$ that is accepted by at least $f + 1$ correct nodes, then \block~$b$ will become promised in round~$i + 1$. Since at least $f + 1$ correct nodes have accepted the \block, however, at least $f + 1$ correct nodes will cast a \permit for the \block~$b$ in round~$i + 1$. Note that these are safe \permits as there is at most one \block created in round~$i$ that may be accepted by at least $f + 1$ correct nodes as there are only $2f + 1$ correct nodes. Similarly, no \proposal may be issued and accepted by at least $f + 1$ correct nodes in round~$i$.
\end{proof}

\subsubsection{Safety of Proposals}
\label{sec:safety:proposals}
Next, let us have a closer look at the role of a \proposal in the \prot algorithm. In essence, in a safe round, a \proposal is a certificate that no promised block may be excluded from the \proposal's position. Hence, we show that it is always safe for any correct node to accept a \proposal.

As a preliminary step, we show that a \proposal's position respects all blocks that could have become promised so far, assuming that the correct nodes properly maintained the safety invariant of the \prot algorithm thus far.
\begin{lemma}
    \label{lem:safety:proposals-encompass}
    If round~$i$ is safe, a \proposal in round~$i$ respects all promised blocks.
\end{lemma}
\begin{proof}
    A \proposal must include $2f + 1$ \permits. At most $\;n - (2f + 1) = f\;$ correct nodes' \permits could have been excluded from the \proposal. Hence, the \proposal must include at least one of the $f + 1$ correct nodes' safe \permits and thus respect all promised blocks.
\end{proof}
Subsequently, we show that the safety invariant is maintained if at least $f + 1$ correct nodes accept the same \proposal.
\begin{lemma}
    \label{lem:safety:safely-accept-proposal}
    If round~$i$ is safe and a \proposal is accepted by at least $f + 1$ correct nodes, then round~$i + 1$ is safe.
\end{lemma}
\begin{remark} 
    For simplicity we assume that a correct node that does not actively see a \proposal, but for arbitrary reasons will issue a \permit for the same position in round~$i + 1$, is also considered to have accepted the \proposal (implicitly).
\end{remark}
\begin{proof} 
    If there is a \proposal issued in round~$i$ that is accepted by at least $f + 1$ correct nodes, the block(s) at the associated position~$p$ will become promised latest in round~$i + 1$. Since at least $f + 1$ correct nodes have accepted the \proposal, however, at least $f + 1$ correct nodes will issue a \permit for the position~$p$ in round~$i + 1$. Note that the position~$p$ is safe as it respects all previously promised blocks (\Cref{lem:safety:proposals-encompass}) and there is at most one \proposal issued in round~$i$ that may be accepted by at least $f + 1$ correct nodes as there are only $2f + 1$ correct nodes. Similarly, no \block may be created and accepted by at least $f + 1$ correct nodes in round~$i$.
\end{proof}

\subsubsection{Safety of other Rounds}
Finally, let us consider the case when no \block or \proposal is accepted by at least $f + 1$ correct nodes in some round. To begin with, we observe that no new \block may become promised in the subsequent round:
\begin{lemma}
    \label{lem:safety:combined-no-newly-promised}
    If round~$i$ is safe and there is no \block or \proposal that is accepted by at least $f + 1$ correct nodes, then no new \block may become promised by the \permits in round~$i + 1$.
\end{lemma}
\begin{proof}
    We prove the \namecref{lem:safety:combined-no-newly-promised} by demonstrating that there is no position receiving at least $f + 1$ correct \permits for the first time in round~$i + 1$. Any new \block is accepted by less than $f + 1$ correct nodes and there cannot be any other correct \permit for such new blocks in round~$i + 1$. Otherwise, correct nodes only update their \current position according to a received \proposal. Assuming that there is no \proposal that is accepted by at least $f + 1$ correct nodes, that is, whose position is adopted by at least $f + 1$ correct nodes, there cannot be a position receiving at least $f + 1$ correct \permits.
\end{proof}
Thus, it remains to show that all the actions taken by correct nodes in such a round will maintain the safety invariant.
\begin{lemma}
    \label{lem:safety:safely-accept-combined}
    If round~$i$ is safe and there is no \block or \proposal that is accepted by at least $f + 1$ correct nodes, then round~$i + 1$ is safe.
\end{lemma}
\begin{proof}
    Following \Cref{lem:safety:combined-no-newly-promised},  it remains to show that at least $f + 1$ correct nodes' \permits remain safe \permits. Without loss of generality, let $v$ be one of the $f + 1$ correct nodes who cast a safe \permit in round~$i$.
    
    
    On the one hand, observe that each correct node transitioning to the next round without updating its \current position would only cast a \permit for the same position again. In other words, $v$ would only cast a safe \permit as its \permit still respects all previously promised blocks and there are no newly promised blocks in round~$i + 1$.
    
    On the other hand, a correct node would only update its \current position by accepting a \block or a \proposal
    . In any case, the updated \current position of a correct node respects all previously promised blocks (\Cref{lem:safety:blocks-respect-promised,lem:safety:proposals-encompass}). Hence, as there are no newly promised blocks in round~$i + 1$, we conclude that the correct node~$v$ would only cast a safe \permit for round~$i + 1$.
\end{proof}

\subsubsection{Safety Invariant}
Ultimately, we may combine the previous statements to conclude \Cref{thm:safety:safety-invariant}, namely that every round is a safe round:
\begin{theorem}
    \label{thm:safety:safety-invariant}
    In any round, a majority of correct nodes would only issue a safe \permit.
\end{theorem}
\begin{proof}
    We show this statement by induction on the round number~$i$. For $i = 0$, observe that the statement holds as all correct nodes start round 0 by voting for the \texttt{genesis} \block, which is also the only \block that may become promised in this round. From now on, assume that some round $i \geq 0$ is safe. We thus have to show that round~$i + 1$ is also safe.
    
    If there is a \block or \proposal that is accepted by at least $f + 1$ correct nodes, then round~$i + 1$ is safe by \Cref{lem:safety:safely-accept-block,lem:safety:safely-accept-proposal}. Otherwise, \Cref{lem:safety:safely-accept-combined} yields the desired statement.
\end{proof}
In other words, in every round, at least $f + 1$ correct nodes' \permits ensure that all promised blocks are contained in a newly promised position.

In the following, we demonstrate what implications \Cref{thm:safety:safety-invariant} imposes on the set of promised blocks:
\begin{definition}
    We say that two blocks $b_1$ and $b_2$ are \textit{conflicting} if neither $b_1$ respects $b_2$, nor $b_2$ respects $b_1$.
\end{definition}
\begin{definition}
    We say that two blocks $b_1$ and $b_2$ become \textit{promised independently} if there are two separate positions, position $p_1$ respecting $b_1$ but not $b_2$ and position $p_2$ respecting $b_2$ but not $b_1$, such that at least $f + 1$ correct \permits were issued for $p_1$ and $p_2$  in some rounds $i$ and $j$.
    
    Furthermore, we say that two blocks become \textit{committed independently}, if there are two such positions where blocks were created.
\end{definition}
\begin{lemma}
    \label{lem:safety:no-two-conflicting-blocks}
    Two conflicting blocks cannot become promised independently.
\end{lemma}
\begin{proof}
    We prove the statement by contradiction. Without loss of generality, assume that \block~$b_i$ becomes promised in round~$i$ and \block~$b_j$ in round~$j$, where $i \leq j$. By \Cref{thm:safety:safety-invariant}, however, at least $f + 1$ correct nodes would only cast a safe \permit in round~$j$. Thus, at least $f + 1$ correct nodes would issue a \permit for a position respecting \block~$b_i$ in round~$j$. As there are only $2f + 1$ correct nodes, there are at most $2f + 1 - (f + 1) = f$ correct nodes who issue a \permit for a position respecting the \block~$b_j$ but not the \block~$b_i$, which is not sufficient for the \block~$b_j$ to become promised in round~$j$. This contradicts the assumption that two independently promised blocks exist.
\end{proof}
The pairwise application of \Cref{lem:safety:no-two-conflicting-blocks} on a set of conflicting blocks in combination with \Cref{lem:protocol:committed-only-if-promised} yields the following \namecref{cor:safety:only-one-conflicting-block}.
\begin{corollary}
    \label{cor:safety:only-one-conflicting-block}
    Out of a set of conflicting blocks, only one block may be independently committed by a child block.
\end{corollary}

\subsubsection{Transactions with Conflicts}
\label{sec:safety:transactions-with-conflicts}
Recall that each \block comes with a set of transactions. As outlined in \Cref{sec:protocol:transactions}, we assume a UTXO system where conflicting transactions are, without exceptions, proof of malicious behaviour.

Following \Cref{def:commit-rule:transactions}, a transaction is said to be \textit{committed} if it is contained in a committed \block and there exists a child \block such that there are no conflicting transactions within the set of committed blocks that are respected by the child \block.

We establish the transaction safety property mentioned above and thereby guarantee that double-spending is impossible. To that end, we argue for the given transaction commit rule that it applies to at most one transaction out of a set of conflicting transactions.
\begin{theorem}[safety]
    Two conflicting transactions cannot become committed simultaneously.
\end{theorem}
\begin{proof}
    To begin with, note that two conflicting transactions may only become committed simultaneously as part of independently committed blocks. In other words, two conflicting transactions cannot be committed jointly within a single \block or within a set of blocks that become committed as a position together. To see this, recall that the commit rule explicitly requires that there is no conflicting transaction within the set of respected blocks of some child \block.
    
    We assume for contradiction that two conflicting transactions may become committed simultaneously within independently committed blocks. Hence, either the blocks are conflicting or not. If the blocks are conflicting, \Cref{cor:safety:only-one-conflicting-block} guarantees that at most one such conflicting \block may be committed independently. If the blocks are not conflicting, then one of the blocks respects the other and thus the transaction commit rule only applies to the transaction within the respected \block.
\end{proof}

\subsection{Liveness}
\label{sec:liveness}
By \textit{liveness}, we denote the property that a distributed system is guaranteed to make progress as soon as the network reaches a state satisfying certain timing assumptions. More specifically, to \textit{make progress} means that outstanding transactions will be included in the ledger and irrevokably committed.

\subsubsection{Timing}
\label{sec:liveness:timing}
First, note that we have previously analyzed the safety of the \prot algorithm without the need to consider any timing information (i.e.\ in the asynchronous model). For \prot's liveness, however, timing is important. For the purpose of this analysis, we assume a global perspective with full information about the state of the entire system
at any considered time~$t$.

Throughout this \namecref{sec:liveness}, we assume the \textit{partially synchronous model} \cite{dwork1988consensus}, that is, we assume that phases of synchronous operation do occur.
During a \textit{phase of synchronous operation}, we assume that messages among correct nodes will be received within a known maximum message delay $\Delta$. Furthermore, we assume that local timers can reliably count the maximum message delays $\Delta$. In particular, we assume that
\begin{equation}
    \label{eq:timeouts}
    2 \Delta \; < \; \texttt{creator timeout} \; < \; 3 \Delta \qquad \text{and} \qquad 5 \Delta \; < \; \texttt{round timeout}
\end{equation}
hold reliably for all correct nodes, where \texttt{ creator timeout } and \texttt{ round timeout } are the actual times for the respective local timeouts to occur. 
Local computation is considered to be negligible.

\subsubsection{Liveness Guarantees}
In \Cref{sec:basic-liveness}, we demonstrate the following basic liveness guarantee for the \prot algorithm:
\begin{restated}{theorem}{restate:block-liveness}
    \label{thm:liveness:block-liveness}
    Assuming that sufficiently long 
    periods of synchrony do occur, the \prot algorithm will create and append new blocks to the ledger.
\end{restated}
At this point, it remains to conclude that new transactions will eventually be included in the ledger and irrevokably committed.
\begin{theorem}[liveness]
    A transaction without conflicts will eventually become executed
    by all correct nodes.
\end{theorem}
\begin{proof}
    A correct creator includes all collected non-conflicting transactions in any block it creates
    . Hence, following \Cref{thm:liveness:block-liveness}, a non-conflicting transaction will be included in a block during a sufficiently long phase of synchronous operation. Furthermore, the created block will be received and accepted by all correct nodes before their \texttt{round timeouts} expire. All correct nodes will subsequently issue \permits to append new blocks from the \block containing the non-conflicting transaction. Following \Cref{lem:safety:blocks-respect-promised} and \Cref{thm:safety:safety-invariant}, all future blocks will thus extend from the \block containing the non-conflicting transaction. Upon another period of synchronous operation (possibly within the same period), three consecutive correct creators will create another block by \Cref{thm:liveness:block-liveness}. Thereby, the commit rule is satisfied for the non-conflicting transaction and it will be executed by all correct nodes upon receiving the new \block.
\end{proof}

\section{Message Complexity}
\label{sec:message-complexity}
Optimistically, the \prot algorithm requires the nodes to send $\mathcal{O}(n)$ messages to create a block (see \Cref{fig:message-schematic}). Note that the message lengths can be optimized using threshold cryptography (as suggested in \cite{gueta2019sbft,yin2019hotstuff}). 

In both the cases of a byzantine block creator or a network failure, however, the nodes fall back to a round of any-to-any communication for a synchronization step, incurring a message complexity of $\theta(n^2)$. Observe that this inflated communication pattern is mainly used in order to tolerate even a network split over an unbounded period of time. Such extreme network conditions should be very rare in practice. We thus suggest that such drastic measures may be employed only every $n$ 
rounds, while keeping the other rounds synchronized based on the local clocks. We thus expect a resulting communication complexity of $\mathcal{O}(n)$ messages.

Note that in signal processing, e.g.\ GPS localization, low-drift oscillators with $\pm0.5\,ppm$ (less than $0.0001\,\%$ or $1 \frac{\mu s}{s}$) are common \cite{oscillator-datasheet}. Thus, even with hundreds of nodes in the system, it should be possible to maintain a sequence of unified rounds over a course of $n$ rounds by increasing the maximum message delay by a small factor in each round after a synchronization step. The careful analysis and implementation of such optimizations, however, lies outside the scope of this paper and is considered future work.













\bibliography{references}

\appendix

\newpage
\section{Proof of the Basic Liveness Guarantee}
\label{sec:basic-liveness}
To show a basic liveness guarantee, we show that three consecutive correct creators are guaranteed to make progress during a phase of synchrony. Subsequently, we argue that there must exist at least one such constellation with three consecutive correct creators and deduce a basic liveness guarantee in the form of \Cref{thm:liveness:block-liveness}.

\subsection{Round Synchronization}
To begin with, we introduce some terminology to describe the possible states of the system over time.
\begin{definition}
    A node~$u$ is said to be in round $r_u(t)$ at time~$t$, that is, the value of its \textnormal{\texttt{round}} variable at time~$t$. Furthermore, let
    $$r_{max}(t) := \max_{u\,\in\,\texttt{correct}} r_u(t)$$
    be the maximum round of any correct node.
\end{definition}
\begin{definition}
    A round~$i$ is \textit{started}, once the first correct node has started the execution of round~$i$. More formally, round~$i$ is started at time~$T_i\,$, where
    $$T_i := \min\{t \mid r_{max}(t) \geq i\}.$$
    We say that a round~$i$ is \textit{unified}, if all correct nodes start the execution of that round in the interval $[T_i, T_i + \Delta]$.
\end{definition}

Next, we demonstrate that any correct creator will sufficiently synchronize the round numbers of all correct nodes.
\begin{lemma}
    \label{lem:liveness:correct-unify-round}
    Assuming synchrony, starting a round~$i$ with a correct creator~$u$ implies that round~$i + 1$ is unified.
\end{lemma}
\begin{proof}
    Observe that round $i = r_{max}(t)$. Round~$i + 1$ starts at time~$T_{i+1}$ when the first correct node~$v$ receives a \block or a \proposal 
    or $2f + 1$ \timeout messages 
    for round~$i$.
    
    Suppose that $v$ has received a \block or \proposal first. Given that the creator~$u$ is correct and the only node capable of producing a \block or \proposal in round~$i$, $u$ must have created the \block or \proposal. As $u$ is correct, it must also be the first correct node to receive the created \block or \proposal (without network delay), so $v = u$. This immediately determines~$T_{i+1}$.
    
    
    Otherwise, suppose that $v$ receives $2f + 1$ \timeout messages first. Again, $v$ immediately transitions to round~$i + 1$ and thus determines~$T_{i+1}$.
    
    In any case, $v$ will immediately broadcast the created \block or \proposal 
    or the $2f + 1$ collected \timeout messages 
    to all nodes at time $T_{i + 1}$. Due to synchrony, each correct node is guaranteed to receive the \block, \proposal or $2f + 1$ \timeout messages within the message delay~$\Delta$ and will immediately transition to round~$i + 1$, no later than time $T_{i + 1} + \Delta$.
    %
\end{proof}
Note that not every correct creator~$u$ necessarily reaches round~$i$ at all. In the case that $u$ does not ever reach round~$i$, round~$i + 1$ would be started as soon as some other correct node~$v$ receives $2f + 1$ \timeout messages for round~$i$.

\subsection{Making Progress}
Subsequently, we show that a correct creator of a unified round will always create a \block or issue a \proposal during a phase of synchronous operation, independent of the correct nodes' \current positions.
\begin{lemma}
    \label{lem:liveness:block-or-proposal}
    Assuming synchrony, a correct creator~$u$ will create a \block or issue a \proposal in a unified round~$i$.
\end{lemma}
\begin{proof}
    In a unified round~$i$, all correct nodes start the execution of the round within the interval $[T_i, T_i + \Delta]$. Hence, the correct nodes will send their \permit to the creator~$u$ 
    for round~$i$ latest at time $T_i + \Delta$. Assuming synchrony, the correct creator~$u$ is thus guaranteed to have received at least $2f + 1$ permits at time $T_i + 2 \Delta$.
    
    Given that the creator~$u$ did not begin round~$i$ before the time~$T_i$, it must have started its \texttt{creator timeout} earliest at time~$T_i$. Consequently, the creator~$u$ receives all permits that arrive until $T_i + \texttt{creator timeout}$. Following \Cref{eq:timeouts}, we have
    $$T_i + 2 \Delta \; < \; T_i + \texttt{creator timeout}\,.$$
    Thus, by the time the \texttt{creator timeout} expires, the creator~$u$ will either have created a \block already, or issue a \proposal now as it received at least $2f + 1$ \permits from the correct nodes.
\end{proof}
If all correct nodes issue \permits for the same position, we observe that a correct creator will create a \block.
\begin{corollary}
    \label{cor:liveness:block-if-unanimity}
    Assuming synchrony, a correct creator~$u$ will create a \block in a unified round~$i$ if all correct nodes issue \permits for the same position.
\end{corollary}
In order to guarantee progress on the ledger, we confirm that any \block or \proposal created by a correct creator in a unified round will be accepted by all correct nodes before each of them transitions to the next round.
\begin{lemma}
    \label{lem:liveness:block-or-proposal-received}
    Assuming synchrony, a \block or \proposal created by a correct creator in a unified round~$i$ will be received and accepted as the new \current position by all correct nodes before starting to execute the next round~$i + 1$.
\end{lemma}
\begin{proof}
    Recall from the proof of \Cref{lem:liveness:block-or-proposal} that a correct creator~$u$ of a unified round~$i$ will have created a \block or issued a \proposal (latest) by the time its \texttt{creator timeout} expires.
    
    Given that round~$i$ is a unified round, the creator~$u$ has started its \texttt{creator timeout} latest at time~$T_i + \Delta$. Consequently, the \texttt{creator timeout} is guaranteed to expire before $T_i + 4 \Delta$ (following \Cref{eq:timeouts}). Now, the correct creator~$u$ immediately sends the created \block or \proposal to all nodes
    . All correct nodes thus receive the \block or \proposal latest at time $T_i + 5 \Delta$ due to the assumed synchrony of the network among correct nodes.
    
    By definition of $T_i$, we know that no correct node has started their local \texttt{round timeout} before time~$T_i$. Following \Cref{eq:timeouts}, we have
    $$T_i + 5 \Delta \; < \; \texttt{round timeout}\,.$$
    Consequently, no correct node's \texttt{round timeout} has expired at time $T_i + 5 \Delta$ and, therefore, no correct node would have created a \timeout message at this point. Since the byzantine nodes can create at most $f$ \timeout messages, no correct node would have received $2f + 1$ \timeout messages. Thus, all correct nodes will receive and accept the created \block or \proposal 
    before starting to execute the next round~$i+1$.
    
    Note that a correct node would fast-forward to a round~$j + 1 > i$, if it received a \block, a \proposal, or $2f + 1$ \timeout messages for round~$j$
    . On the one hand, no such round~$j \geq i$ could start before round~$i$ (i.e.\ $T_j \geq T_i$); hence, no correct node would have created a \timeout message for round~$j$ by the time~$T_i + 5 \Delta$.
    
    On the other hand, any \block or \proposal for round~$j - 1$ would have to be created using the \permits received from at least $f + 1$ correct nodes. However, given that there cannot exist sufficiently many \timeout messages to abort round~$i$, these $f + 1$ correct nodes must have already received and accepted the \block or \proposal of round~$i$. In consequence, a \block or \proposal created in a following round must either extend from the \block created in round~$i$, or respect the position of the \proposal issued in round~$i$, respectively. In any case, a correct node accepting such a \block or \proposal to fast-forward to round~$j + 1$ instead, will implicitly accept the \block or \proposal of round~$i$.
\end{proof}
Before we establish the basic liveness guarantee for the \prot algorithm, we observe one final property:
\begin{lemma}
    \label{lem:liveness:3-consecutive}
    There exists at least one sequence of three consecutive correct creators within $n + 2$ consecutive rounds.
\end{lemma}
\begin{proof}
    We have $n = 3f + 1$ nodes, whereof $f < \frac{n}{3}$ are byzantine. Recall that we choose the creator of a round in a round-robin (wrap-around) scheme. By arranging the nodes in a circle, we may pick the first of the three consecutive correct creators among the next $n$ nodes. The third consecutive correct creator will be reached within at most $n + 2$ rounds. It remains to show that one iteration of the circle contains three consecutive correct creators.
    
    To that end, consider any fixed order of the $2f + 1$ correct creators in a circle. The correct creators can be divided into sequences of consecutive correct creators by placing byzantine nodes in between. Note that, by placing the $f$ byzantine nodes, the correct creators can be partitioned into at most $f$ sequences. Now, if we assume for contradiction that each of these sequences of consecutive creators would have length at most 2, then there could only be $2f < 2f + 1$ correct nodes
    \linebreak
    -- a contradiction.
\end{proof}
Ultimately, we may now conclude \Cref{thm:liveness:block-liveness}:
\restate{restate:block-liveness}
\begin{proof}
    Once a sufficiently long period of synchrony occurs, it follows from \Cref{lem:liveness:3-consecutive} that there must be three consecutive rounds $r_1$, $r_2$ and $r_3$ with correct creators within the next $n + 2$ rounds. A round can be terminated either by creating a \block or \proposal, or with $2f + 1$ \timeout messages. Hence, given that there are $2f + 1$ correct nodes, the round~$r_1$ will be started at some point.
    
    By \Cref{lem:liveness:correct-unify-round}, the rounds $r_2$ and $r_3$ are unified rounds. Furthermore, by \Cref{lem:liveness:block-or-proposal,lem:liveness:block-or-proposal-received}, the creator of round~$r_2$ creates a \block or issues a \proposal that will be received (and accepted as the new \current position) by each correct node before starting to execute the next round~$r_3$. Note that all correct nodes will now issue \permits for the same position in round~$r_3$. Consequently, by \Cref{cor:liveness:block-if-unanimity} and \Cref{lem:liveness:block-or-proposal-received}, the creator of round~$r_3$ is guaranteed to create a \block that is received and accepted by all correct nodes before starting the next round.
\end{proof}

\end{document}